\documentclass[runningheads,orivec]{llncs_manual}

\usepackage[margin=2.5cm]{geometry}
\usepackage[fontsize=11pt]{scrextend}
\usepackage{amsfonts,amsmath,amssymb}
\usepackage{graphicx}
\usepackage[ruled,vlined,linesnumbered]{algorithm2e}
\usepackage{quiver}
\usepackage{complexity}
\usepackage[strings]{underscore}
\usepackage{cite}
\usepackage[T1]{fontenc}
\usepackage{lmodern}
\usepackage{textcomp}
\usepackage{url}
\usepackage{bbold}
\DeclareGraphicsExtensions{.png,.pdf,.svg}
\usepackage{enumerate}

\SetKwInput{KwShared}{Shared}
\SetKwInput{KwInitial}{Initial}
\SetKwInput{KwPattern}{write-read pattern  }
\SetKwInput{KwImmediately}{immediately}

\DeclareMathOperator{\Ima}{Im}

\def\REVIEW{0}  

\if \REVIEW 1
    \newcommand{\guillermo}[1]{{\color{blue}[Guillermo: #1]}}
    \newcommand{\petr}[1]{{\color{red}[Petr: #1]}}
\else
    \newcommand{\guillermo}[1]{}
    \newcommand{\petr}[1]{}
\fi

\def\S{\ensuremath{\mathcal{S}}}

\def\D{\ensuremath{\mathcal{D}}}

\def\A{\ensuremath{\mathcal{A}}}
\def\B{\ensuremath{\mathcal{B}}}

\def\I{\ensuremath{\mathcal{I}}}
\def\O{\ensuremath{\mathcal{O}}}

\def\K{\ensuremath{\mathcal{K}}}

\def\xIIS{\Xi_{I\!I\!S}}

\begin{document}
%
\title{Space-Time Trade-off in Bounded Iterated Memory 
}
%
\author{Guillermo Toyos-Marfurt\inst{1,2} \and Petr Kuznetsov\inst{1}}
%
\authorrunning{G. Toyos-Marfurt \and P. Kuznetsov}
%
\institute{Télécom Paris, Institut Polytechnique de Paris, France \and
Télécom SudParis, Inria Saclay, Institut Polytechnique de Paris, France \\
\email{\{guillermo.toyos,petr.kuznetsov\}@telecom-paris.fr}}
\maketitle              
\begin{abstract}
The celebrated asynchronous computability theorem (ACT) characterizes tasks solvable in the read-write shared-memory model using the \emph{unbounded full-information protocol}, where in every round of computation, each process shares its complete knowledge of the system with the other processes. 
Therefore, ACT assumes shared-memory variables of unbounded capacity. 
It has been recently shown that \emph{bounded} variables can achieve the same computational power at the expense of extra rounds.
However, the exact relationship between the bit capacity of the shared memory and the number of rounds required in order to implement one round of the full-information protocol remained unknown. 

In this paper, we focus on the asymptotic \emph{round complexity} of bounded iterated shared-memory algorithms that simulate, up to isomorphism, the unbounded full-information protocol.
We relate the round complexity to the number of processes $n$, the number of iterations of the full information protocol $r$, and the bit size per shared-memory entry $b$.
By analyzing the corresponding \emph{protocol complex}, a combinatorial structure representing reachable states, we derive necessary conditions and present a \emph{bounded} full-information algorithm tailored to the bits available $b$ per shared memory entry.
We show that for $n>2$, the round complexity required to implement the full-information protocol satisfies $\Omega((n!)^{r-1} \cdot 2^{n-b})$.
Our results apply to a range of iterated shared-memory models, from regular read-write registers to atomic and immediate snapshots. 
Moreover, our bounded full-information algorithm is asymptotically optimal for the iterated collect model and within a linear factor $n$ of optimal for the snapshot-based models.


\end{abstract}

\section{Introduction}
\label{sec:intro}
In distributed systems, one of the key characteristics is the system’s computational power, that is, the range of tasks it allows for solving. 
Models of distributed computing are defined by a plethora of parameters and their comparative task computability analysis remains a significant challenge.
The celebrated Asynchronous Computability Theorem (ACT)~\cite{HS99} characterizes wait-free task computability for read-write shared-memory models. 

Combinatorial topology~\cite{distCompTopo} has been proven essential in characterizing such iterative shared-memory models~\cite{HS99,GKM14-podc,KRH18}.
Specifically, the set of executions of a given model can be represented as a \emph{simplicial complex}: a combinatorial structure that enables analysis of the model's key properties.
In this context, a distributed \emph{task} is defined as a tuple $(\I,\O, \Delta)$, where $\I$ and $\O$ represent the \emph{input} and \emph{output} complexes, defining the task’s inputs and outputs, and $\Delta: \I \rightarrow 2^{\O}$ is a mapping that associates each input assignment with the set of \emph{allowed} output assignments. 

ACT~\cite{HS99} states that a task is \textit{solvable} if and only if there exists a \emph{chromatic} map from $\I$ to $\O$ that \emph{respects} $\Delta$. \footnote{In the special case when the task is \emph{colorless}, the characterization boils down to the existence of a continuous map from $|\I|$ to $|\O|$ (geometric realizations of $\I$ and $\O$).}
This characterization assumes the \emph{full-information protocol}, in which the processes share their \emph{complete} (rapidly growing) state in every round of computation, which relies on shared-memory variables of \emph{unbounded} size.

It has been proved that, strictly speaking, in \emph{wait-free} read-write iterated models, unbounded memory is not necessary for preserving computational power~\cite{delporte2024computational}. 
Indeed, variables of minimal (one-bit) capacity can be used to implement an iteration of the full-information protocol, at the expense of running through multiple rounds.  
But this raises a fundamental question:
What is the cost of bounded memory? 
Specifically, while bounded memory may be theoretically sufficient for simulating full-information, what practical costs, such as round complexity, are required to achieve the same computing power?

This paper characterizes a family of iterated shared-memory algorithms~\cite{rajsbaum2010iterated}, called $\mathit{ITER}$, that implement the full-information protocol. 
We asymptotically determine the necessary number of rounds an $\mathit{ITER}$ algorithm should go through so that its protocol complex becomes isomorphic to $r$ rounds of the full-information protocol.

Our approach uses the concept of process \emph{distinguishability} to map shared memory values in a way that transmits full-information within bounded bits. 
By mapping shared memory entries to states of the full-information protocol, we derive the number of rounds required to emulate $r$ iterations of the full-information protocol, as a function of the available registers and the number of processes $n$.

Our contributions are as follows. 
For a given input complex $\I$, we establish the relationship between round complexity and available bits $b$ per shared-memory entry required to simulate $r$ iterations of a full-information protocol. 
Specifically, for systems with three or more processes, $(n>2)$, we show that the round complexity (i.e., the asymptotic number of iterations required by the algorithm) has a lower bound of  $\Omega( {(n!)^{r-1} \cdot 2^{n-b}})$.
For two processes, the round complexity is $1$, achievable with just $2$-bit registers~\cite{prequel}. 
We also present a construction to obtain bounded full-information algorithms that are asymptotically optimal in round complexity for the iterated collect model, and execute at most $O( {(n!)^{r-1} \cdot 2^{n-b} \cdot n})$ rounds in the atomic and immediate snapshot models.
These results apply to a broad class of iterated shared-memory algorithms, $\mathit{ITER}$, which include well-known iterated models such as the Iterated Immediate Snapshot (IIS)~\cite{BG93a}, Iterated Atomic Snapshot (IAS)~\cite{10.5555/2821576,bouzidStrg}. and Iterated collect (IC)~\cite{DBLP:books/daglib/0017536}.

\subsubsection{Roadmap.}
The paper is organized as follows. 
We overview the related work in Section~\ref{sec:related}. 
In Section~\ref{sec:Preliminaries}, we recall the combinatorial topology tools and concepts used throughout this work.
In Section~\ref{sec:sys_model}, we describe the iterated family of algorithms $\mathit{ITER}$.
Then in Section~\ref{sec:lower_condition}, we give the necessary conditions that algorithms in $\mathit{ITER}$ must satisfy for having a protocol complex equivalent to the full-information protocol.
Next, in Section~\ref{sec:greedy_star}, we introduce Greedy Star, a construction for obtaining bounded full-information protocols based on the combinatorial condition of Section~\ref{sec:lower_condition}.
Finally, in Section~\ref{sec:conclusions}, we discuss the final results.

\section{Related work}
\label{sec:related}

This paper builds on the research started in~\cite{prequel}, and further extended in~\cite{online_version}, which characterized the iterated immediate snapshot model (IIS)~\cite{BG93a,GR10-opodis} by analyzing its protocol complex, which is equivalent to the iterated standard chromatic subdivision~\cite{chromsubdivision}. 
The authors determine the necessary and sufficient conditions on the amount of shared-memory to simulate the full-information protocol \emph{in the same number of rounds} as its unbounded counterpart.
Our work extends this line of research by allowing an iterated algorithm to perform multiple rounds for a single iteration of the unbounded protocol, examining the relationship between round complexity and bit complexity.
We further generalize these ideas to a broader class of protocols, which we refer to as $\textit{ITER}$.

The Asynchronous Computability Theorem (ACT), introduced by Herlihy and Shavit~\cite{HS99}, showed the impossibility of wait-free set agreement, solving the long-standing problem at the time~\cite{SZ93,BG93a}.
Herlihy, Kozlov and Rajsbaum~\cite{distCompTopo} offered a comprehensive treatment of combinatorial topology as a tool in distributed computing. 
Hoest and Shavit~\cite{hoest} proposed an alternative IIS model to study its execution time complexity.

Delporte-Gallet et. al.~\cite{2process_complexity} studied the \emph{wait-free} solvability of tasks using $1$-bit messages in dynamic networks.
Multiple works considered iterated algorithms in more general model classes. 
Bouzid et. al.~\cite{bouzidStrg} established an equivalence relationships between iterated and non-iterated read-write memory adversarial models.
Kuznetsov et al.~\cite{KRH18} proposed a characterization of task solvability in a large class of adversarial iterated models.  

Recently, Delporte-Gallet et. al.~\cite{delporte2024computational} characterized the computational power of bounded registers in the wait-free and $t$-resilient models. 
Their focus was on (task) \emph{computability}, not addressing the relationship between bounded shared-memory size and \emph{round complexity}.
In this paper, we fill this gap.

\section{Preliminaries}
\label{sec:Preliminaries}
In this section, we briefly review the key definitions and concepts, introduced in~\cite{prequel}, which are essential for this work.

\subsubsection{Simplicial Complex.} 
We represent the states of iterated shared-memory protocols as topological-combinatorial spaces, called \emph{simplicial complexes}~\cite{distCompTopo}: 
a set of \emph{vertices} and an inclusion-closed set of vertex subsets, called \emph{simplices}. 
We call the simplicial complex of $n+1$ vertices with its power set the $n$-dimensional simplex $\Delta^n$.
The leftmost triangle in Figure~\ref{fig:subdivisions} corresponds to $\Delta^2$.
The \emph{dimension} of a simplex $\Delta$, denoted as $dim(\Delta)$, is its number of vertices $V(\Delta)-1$.
The dimension of a simplicial complex is equal to the dimension of the largest simplex it contains.
We call \textit{faces} the simplices contained in a simplicial complex, where a $0$-face is a vertex.
The set $\textit{Faces}(\A)$ consists of all faces of the simplicial complex $\A$.
As an example, $\textit{Faces}(\Delta^2)$ has three 0-faces (vertices), three 1-faces (edges), and one 2-face (the simplex itself).
We call \textit{facets} the simplices which are not contained in any other simplex. 
We denote the set of vertices of a simplicial complex $\mathcal{A}$ as $V(\mathcal{A})$.
Two vertices are \emph{adjacent} if there is a $1$-face (an ``edge'') in the simplicial complex containing both vertices.
We denote $\mathrm{deg}(\A,v)$ the number of adjacent vertices to a vertex $v$ in $\A$.
We call a simplicial complex \emph{chromatic} when it comes with a \emph{coloring function} $\pi : V(\A) \rightarrow \Pi$ that assigns every vertex to a process identifier. We consider that all simplicial complexes are chromatic. 
Given a simplicial complex $\mathcal{A}$, the \emph{star} of $\mathcal{S}\subseteq\mathcal{A}$, $\mathrm{St}(\mathcal{A},\mathcal{S})$, is a subcomplex made of all simplices in $\mathcal{A}$ containing a simplex of $\mathcal{S}$ as face. 

\subsubsection{Subdivisions.}
In general, a subdivision operator $\tau : \A \rightarrow \B$ is a map that ``divides'' the simplices in $\A$ into smaller simplices. For a rigorous definition see~\cite{distCompTopo}.  
We say that a simplicial complex $\B$ subdivides $\A$ if there exists a subdivision operator $\tau$ such that $\B = \tau(\A)$. 
Subdivisions are intersection preserving: for any subcomplexes $\A,\B \subseteq \I$, $\tau(\A) \cap \tau(\B) = \tau(\A \cap \B)$.
The \emph{standard chromatic subdivision} of a simplicial complex $\mathcal{A}$, denoted as $\mathrm{Ch}\ \mathcal{A}$, is a complex whose vertices are tuples $(c,\sigma)$ where $c$ is a color and $\sigma$ is a face of $\mathcal{A}$ containing a vertex of color $c$. 
A set of vertices of $\mathrm{Ch}\ \mathcal{A}$ defines a simplex if for each pair $(c,\sigma)$ and $(c',\sigma')$, $c\neq c'$ and either $\sigma \subseteq \sigma'$ or $\sigma' \subseteq \sigma$. 
For $\Delta^n$, the vertices $(c,\Delta^n)$ define a \emph{central simplex} in $\mathrm{Ch}\ \mathcal{A}$, these vertices are called \emph{central vertices}.
The standard chromatic subdivision can be seen as the ``colored'' analog of the \emph{standard barycentric subdivision}~\cite{distCompTopo}. 
Note that $\mathrm{Ch}$ is itself a subdivision operator: for any simplicial complex $\A$, the complex $\mathrm{Ch}\ \mathcal{A}$ subdivides $\mathcal{A}$.
For instance, if $\mathcal{A} = \Delta^2$, then $\mathrm{Ch}\,\Delta^2$ subdivides
$\Delta^2$ as shown in Figure~\ref{fig:subdivisions} (top right).

In this paper, we characterize a necessary condition for full-information protocols using the combinatorial growth of vertex degrees under iterative chromatic subdivisions of an input complex.
Specifically, we use Theorem 5 of~\cite{online_version}, originally stated in terms of the number of faces in the local neighborhood of a vertex. 
In our setting, we consider the case corresponding to the number of $1$-dimensional faces adjacent to a vertex, that is, its degree.

\begin{lemma}[Asymptotic bound on the iterative chromatic subdivision~\cite{online_version}]\label{lemma:asymp_boundf}
Let $r>1$, $\K$ a $n$-dimensional simplicial complex and $v \in V(\K)$. The following is a tight asymptotic bound on the degree of $v$ in $\mathrm{Ch}^r\ \K$:
\[
 \mathrm{deg}(\mathrm{Ch}^r\ \K,v) \in \Theta\bigg((n!)^{r-1} 2^n n\bigg)
\]
\end{lemma}

\subsubsection{Tasks.} 
A \emph{distributed task} is defined as a tuple $(\I,\O, \Delta)$, where $\I$ and $\O$ are, resp., the \emph{input complex} and the \emph{output complex}, describing the task's input and output configurations, and $\Delta: \I \rightarrow 2^{\O}$ is a map relating each input assignment to all output assignments allowed by the task.
In this work we consider general \emph{colored} tasks, where $\I$ and $\O$ are chromatic simplicial complexes, allowing processes to have different sets of inputs and outputs.
To solve a task, a distributed system uses a (communication) \emph{protocol} to share information between processes.
The set of configurations reachable by the protocol from a given initial configuration  are expressed via a carrier map, denoted by $\Xi$.
A carrier map is a function from a simplicial complex to its power set that preserves face inclusion: for all simplices $\tau, \sigma \in \textit{Faces}(\mathcal{I})$ with $\tau \subseteq \sigma$, it holds that $\Xi(\tau) \subseteq \Xi(\sigma)$. We refer to this function as the \emph{protocol map}.
From the protocol complex $\Xi(\mathcal{I})$, a 
simplicial map (i.e. a function that maps simplices to simplices) determines the outputs of processes by mapping the simplices of the protocol complex to an output complex, respecting the specification of the problem given by $\Delta$: $\delta \circ \Xi(\mathcal{I}) \subseteq \Delta(\mathcal{I})$. We refer to $\delta$ as the \emph{decision function}.
Note that, unlike linearizability~\cite{HW90}, correctness of a distributed task is expressed topologically: a protocol is correct if it maps every input simplex to an allowed output simplex respecting the task specification $\Delta$.

\section{System Model}
\label{sec:sys_model}

\subsection{Iterated Memory}
\label{subsec:computational_moodel}

We introduce $\mathit{ITER}$, a family of iterated memory models.
Algorithm~\ref{alg:IterMemory_algo} shows the pseudo-code of a model in $\mathit{ITER}$. 
The motivation behind this  family is to provide a \emph{generic} representation that covers all iterated read-write shared-memory protocols, allowing us to characterize them.

Computation is structured in a sequence of \textit{shared memory layers} $M[1],\ldots,M[R]$. 
In each layer, there is an array of $n$ SWMR registers,  one register per process, initialized with $\bot$. 
Processes then can write to their dedicated registers and read the registers of the layer.
The $r$-th round of algorithm $A$, denoted $A(r)$, consists of each process performing a write to its register on $M[r]$, reading all registers of the layer, and updating its local state accordingly (i.e., one iteration of the loop in Algorithm~\ref{alg:IterMemory_algo}).

The models in $\mathit{ITER}$ are \emph{wait-free}: each process reaches an output in a finite number of its own steps, regardless of the other processes.
Each process is initially assigned an \emph{input}, and maintains its state in a local variable $s$.
In each iteration, the process writes a representation of its current state using an $\textit{encode}$ function (line~$2$).
Then, it reads the whole memory layer (line~$3$) and changes the internal process state according to a $\textit{next\_state}$ transition function.
Finally, after a fixed number of rounds $R$, the process returns a value according to its task-specific decision function $\delta_i$ according to its internal state.

\begin{algorithm}
\SetAlgoLined
\KwShared{$M[R,|\Pi|]$ array of $R$ arrays of shared memory, each with $|\Pi|$ entries.}
\KwInitial{$v = input(i)$ $\triangleright$ What the process sees. At first, its input.}
\KwInitial{$s \gets next\_state(v,\bot,0)$ $\triangleright$ Initial state of the process.}
\For{$r:=1$ to $R$}{
    \KwPattern{}
    \ \ \ \ $M[r,i]$ $\gets$ $write_P(encode(s,r))$\;
    \ \ \ \ $v$ $\gets$ $read_P(M[r])$\;
    $s \gets next\_state(s,v,r)$
}
\Return $\delta_i(s)$
\caption{The family of iterated memory algorithms. Code for process $p_i \in \Pi$, $R>0$ iterations and write-read pattern $P$.}
\label{alg:IterMemory_algo}
\end{algorithm}


An algorithm $A \in \mathit{ITER}$ is characterized by its $\textit{encode}$ and $\textit{next\_state}$ functions, the access pattern to shared memory it uses in its write-read primitives, the number of iterations $R$ and the decision functions $\delta_i$.
Depending on the implementation of the $\textit{read}_P$ and $\textit{write}_P$ primitives, we can obtain different well-known iterated models present in the distributed systems literature:

\begin{itemize}
  \item \textbf{Iterated Collect (IC):} Each layer is read using a \emph{collect} operation, where each process reads the $n$ registers of the layer one by one in any order~\cite{DBLP:books/daglib/0017536}.
\item \textbf{Iterated Atomic Snapshot (IAS):}  Instead of reading each SWMR register independently, the read implements a snapshot operation which \emph{atomically} reads all registers at the same time~\cite{10.5555/2821576}.
 \item \textbf{Iterated Immediate Snapshot (IIS):}  The memory is accessed using immediate snapshots, where the read and write primitives are executed immediately as if they were a single atomic operation~\cite{petrBook}.
\end{itemize}

In this work, we will focus on the \emph{full-information} protocols~\cite{distCompTopo}: In each layer, the processes try to communicate all their ``knowledge'' by writing their \emph{complete views} of the system. 
The encoding function thus should transmit all the necessary information by representing these views as a sequence of bits. 

Note that if both the encoding and the state function are the identity -- that is, $\textit{encode}: state \mapsto state$ and $\textit{next\_state}: (view,state) \mapsto view$ -- the protocol corresponds to the \emph{unbounded} full-information protocol and the bit complexity is $n^{r}$ in the $r$-th iteration. 
Here, only the \emph{read-write pattern} characterizes the algorithm. 
We denote this particular set of algorithms as $\textit{FI}_P$, where $P$ is the pattern it uses. 
For example, $\textit{FI}_{IIS}$ refers to the well-studied full-information iterated immediate snapshot (IIS) algorithm, $\textit{FI}_{IAS}$ to its atomic snapshot counterpart, and $\textit{FI}_{IC}$ to the iterated collect variant.

However, if we impose restrictions on the size of outputs in the encoding function we can obtain algorithms that can be implemented using \emph{bounded-size} registers. 
For instance, by using the immediate snapshot access pattern and a bounded encoding function, we get the Bounded Iterated Immediate Snapshot (B-IIS) model studied in~\cite{prequel}.
Thus, the objective of this work is to characterize algorithms in $\mathit{ITER}$ that are equivalent to $\textit{FI}_P$ but use bounded memory at each shared-memory layer. 

For brevity, given an algorithm $A \in \mathit{ITER}$, let $\S$ be the set of possible states of the variable $s$, $V$ the set of possible views obtained by the read primitive and $E$ the image of $\textit{encode}$ (also called \emph{encoding set}).
An algorithm $A \in \mathit{ITER}$ is equipped with a sequence of functions $\{\omega_{k}\}_{k \in [R]}$ and $\{\sigma_{k}\}_{k \in [R]}$ such that at round $r$, $\sigma_r \in \{\sigma_{k}\}_{0<k \leq R}, \sigma_r : \S \times V \mapsto \S$ and $\omega_r \in \{\omega_{k}\}_{0<k\leq R}, \omega_r : \S \mapsto E$ are used as the $\textit{encode}$ and $\textit{next\_state}$ functions respectively.


\subsection{Topological Model}

We consider a system  of $n$ asynchronous processes, which can fail arbitrarily by crashing.
Processes communicate using an algorithm $A$ from the family of iterated shared memory algorithms $\mathit{ITER}$ described in section~\ref{subsec:computational_moodel}.

We denote by $\Xi_A$ the protocol map of an algorithm $A \in \mathit{ITER}$. 
Since $A$ uses a different pair of functions $(\omega_r, \sigma_r)$ at each iteration $r$, it induces a distinct intermediate protocol map $\Xi^*_{A(r)}$ for each round $r$.
We write $\Xi_{A(k)}(\I)$ to refer to the protocol complex obtained after applying the sequence of intermediate protocol maps from round $1$ to $k$: $\Xi_{A(k)}(\I) = \Xi^*_{A(k)} \circ \Xi^*_{A(k-1)} \circ \dots \circ \Xi^*_{A(1)}(\I)$.
Accordingly, we have that $\Xi_A = \Xi_{A(R)}$, were $R$ is the last round of $A$ (as defined in Algorithm~\ref{alg:IterMemory_algo}).

Figure~\ref{fig:subdivisions} highlights the resulting protocol maps of the well-known full information protocols that use different shared memory primitives. We look for bounded memory protocol maps that yield the same output complex as the unbounded full information protocols $\textit{FI}_P$.
That is, given a simplicial input complex $\mathcal{I}$, and an algorithm $A \in \mathit{ITER}$ using a read write pattern $P$. We want that the protocol complex $\Xi_A(\mathcal{I})$ to be \emph{isomorphic} to the protocol complex yielded by $\textit{FI}_P$. Particularly, if $A$ uses immediate snapshot, we want that $\Xi_A(\mathcal{I}) \cong \xIIS(\mathcal{I}) \cong \mathrm{Ch}(\mathcal{I})$
\begin{figure}
    \centering
    \includegraphics[width=0.6\columnwidth]{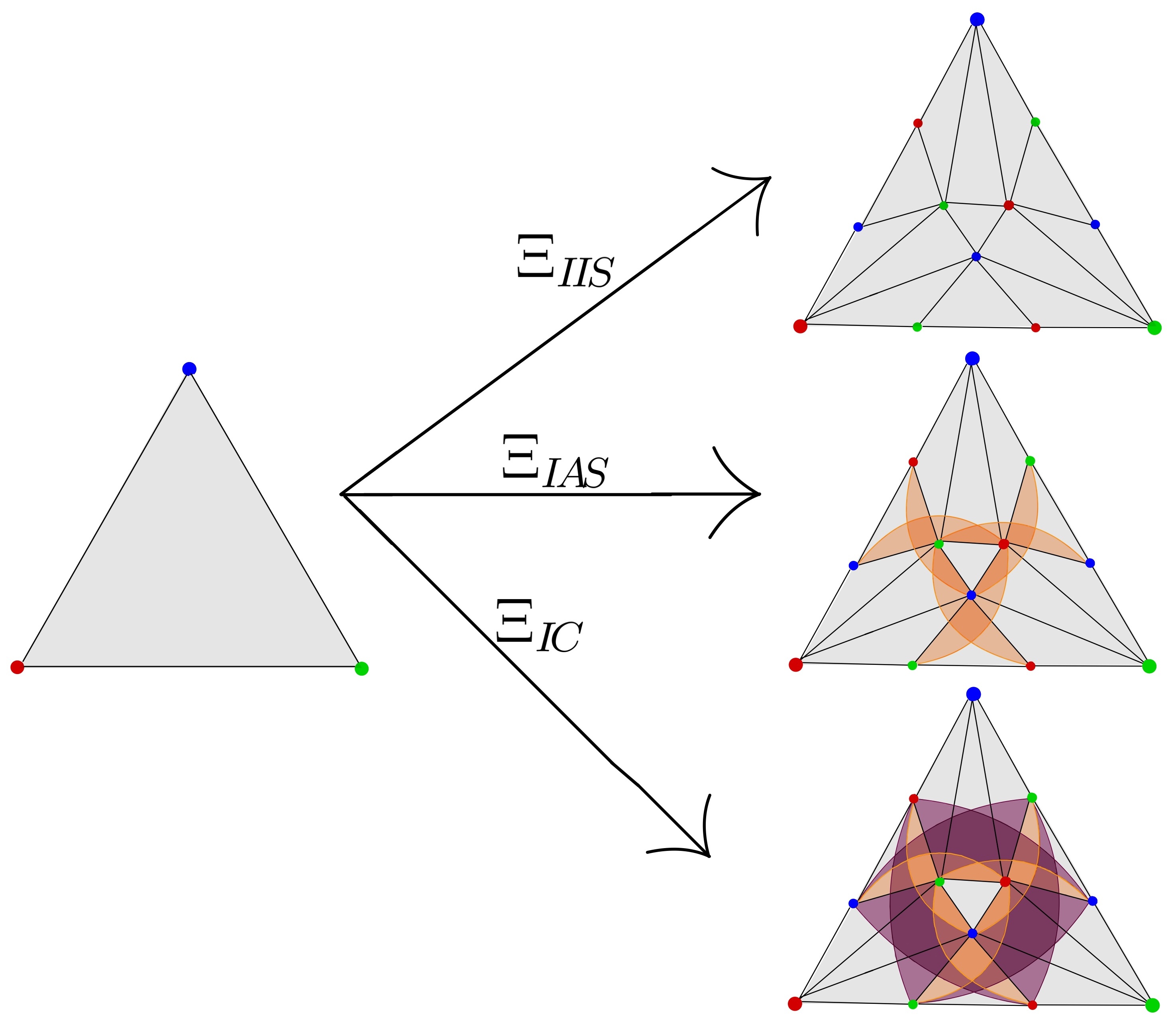}
    \caption{Application of the  \textit{immediate snapshot}, \textit{atomic snapshot} and \textit{iterated collect} protocol maps to $\Delta^2$ (shown on the right, from top to bottom). Starting from $\xIIS(\Delta^2)$,
    the new possible configurations in the resulting protocol complexes of $\Xi_{I\!A\!S}$ and $\Xi_{IC}$ are highlighted in orange and purple, respectively. Note that $\xIIS(\Delta^2) \subset \Xi_{I\!A\!S}(\Delta^2) \subset \Xi_{IC}(\Delta^2)$ and $\xIIS \cong \mathrm{Ch}$.}
    \label{fig:subdivisions}
\end{figure}

\section{A Necessary Condition for Full-Information Protocols}
\label{sec:lower_condition}

In this paper, we aim to identify a condition that an algorithm  $A \in \mathit{ITER}$
%
using a write-read pattern $P$
must satisfy in order to achieve an equivalence, up to isomorphism, with the (unbounded) full-information algorithm  $\textit{FI}_P$. 
More precisely, we derive a necessary condition on the protocol complex $\Xi_A(\I)$, given an input complex $\I$.

\subsection{Properties of $\textit{FI}_P$}

Our first objective is to establish some properties of the full-information protocol map of $\textit{FI}_P$. For example, we know that $\Xi_{\textit{FI}_{\textit{IIS}}}(\I) = \textit{Ch}$ as established in \cite{chromsubdivision}. However, this does not hold for all write-read patterns.

In Figure~\ref{fig:subdivisions}, we depict the protocol maps corresponding to one iteration of the iterated atomic snapshot ($\textit{FI}_{\textit{IAS}}$) and iterated collect ($\textit{FI}_{\textit{IC}}$) in $3$-process systems. 
These protocol complexes do not preserve planarity, implying they are not subdivisions of the input complex. 
Nevertheless, they share some key properties with subdivision operators.
Specifically, Lemma~\ref{lemma:FIP_prop_mesh} and Lemma~\ref{lemma:FIP_prop_cap} demonstrate that for any write-read pattern $P$, $\textit{FI}_P$ induces a mesh-shrinking protocol map~\cite{distCompTopo}. 
Intuitively, this means that in the $k$-th round of the protocol, the shortest path between any pair of vertices strictly increases with each previous round $k'$ such that  $k>k'$.

Therefore, any algorithm A $\in \mathit{ITER}$ that is equivalent, up to isomorphism, to a full-information protocol map must also satisfy the properties stated in Lemmas~\ref{lemma:FIP_prop_mesh} and~\ref{lemma:FIP_prop_cap}. 
These lemmas are applied in the proofs of Lemma~\ref{lemma:lateResultLowerBound} and Theorem~\ref{theorem:encodings}.

\begin{lemma} \label{lemma:FIP_prop_mesh}
    Let $\I$ be an input complex, $\textit{FI}_P$  an unbounded full information algorithm using the write-read pattern $P$ and $\Xi_{\textit{FI}_P}$ its protocol map.
    Then:
    \[
    \forall v,w \in V(\I), \{v,w\} \notin \Xi_{\textit{FI}_P}(\I)
    \]
    \end{lemma}
    \begin{proof}
    By contradiction, suppose there exists a pair of vertices $v,w \in V(\I)$ such that the edge $ \{v,w\} \in \Xi_{\textit{FI}_P}(\I)$.
    That is, there is a global configuration in $\Xi_{\textit{FI}_P}(\I)$ where processes $\pi(v)$ and $\pi(w)$ read nothing but their own input.
    In particular, both processes did not see each other's inputs.
    By definition, $\textit{FI}_P \in \textit{ITER}$. Since it follows the pseudocode of Algorithm~\ref{alg:IterMemory_algo}, we have that all processes finish writing their input before reading the shared memory. 
    Thus, if $\pi(v)$ did not read $\pi(w)$'s input, the read operation of $\pi(v)$ started before $\pi(w)$ finished writing its input. 
    And the write operation of $\pi(v)$ finished before the write operation of $\pi(w)$.
    Conversely, as $\pi(w)$ also did not read $\pi(v)$'s input, the write operation of $\pi(w)$ finished before the write operation of $\pi(v)$, which is a contradiction.  
    $\qed$
    \end{proof}
    
    \begin{lemma} \label{lemma:FIP_prop_cap}
    Let $\I$ be an input complex and $\A,\B \subseteq \I$, $\textit{FI}_P$ an unbounded full information algorithm using the write-read pattern $P$ and $\Xi_{\textit{FI}_P}$ its protocol map:
    \[
        \Xi_{\textit{FI}_P}(\A) \cap \Xi_{\textit{FI}_P}(\B) = \Xi_{\textit{FI}_P}(\A \cap \B)
    \]
    \end{lemma}
    
    \begin{proof}
    
    Let $\sigma \in \Xi_{\textit{FI}_P}(\A) \cap \Xi_{\textit{FI}_P}(\B)$. By definition, $\sigma$ is a final view configuration of $|\sigma|$ processes that can be reached from both $\A$ and $\B$. 
    Since $\sigma \in \Xi_{\textit{FI}_P}(\A)$, the inputs read by processes in $\pi(\sigma)$ must have originated from some initial configuration in $\A$.
    Similarly, as $\sigma \in \Xi_{\textit{FI}_P}(\B)$, the same inputs must also be in $\B$. It follows that the initial input configuration of $\sigma$ is contained in $\A \cap \B$, implying that $\sigma \in \Xi_{\textit{FI}_P}(\A \cap \B)$. 
    Thus, $\Xi_{\textit{FI}_P}(\A) \cap \Xi_{\textit{FI}_P}(\B) \subseteq \Xi_{\textit{FI}_P}(\A \cap \B)$.

    Conversely, let $\sigma \in \Xi_{\textit{FI}_P}(\A \cap \B)$. This means $\sigma$ is a final configuration that results from an initial configuration in $\A \cap \B$.
    Since $\A \cap \B \subseteq \A $ and $ \A \cap \B \subseteq \B$, it follows that $\sigma \in  \Xi_{\textit{FI}_P}(\A)$ and $\sigma \in  \Xi_{\textit{FI}_P}(\B)$. 
    Thus, $\sigma \in  \Xi_{\textit{FI}_P}(\A) \cap \Xi_{\textit{FI}_P}(\B)$. $\qed$
    \end{proof}

\subsection{The Necessary Condition}

We now show which conditions $A \in \textit{ITER}$ must satisfy in order to have the same properties as the unbounded full information protocol $\textit{FI}_P$.
The key concept in achieving this equivalence is \textit{process distinguishability}. 
The goal is for a process to determine the state of another by reading the value it has written in shared memory. 
Importantly, it is not necessary for processes to write all of their knowledge into shared memory; a small, decodable, value may suffice to uniquely identify a state within the protocol complex. This allows us to define encoding functions that use a bounded number of shared-memory bits. 

The following definitions define the notion of distinguishability for subcomplexes.

\begin{definition}[Vertex distinguishability] \label{def:disting}
    Let $\K$ be a simplicial complex and $\omega : V(\I)\rightarrow E$. We say that $v \in V(\K)$ is \textit{distinguishable} in $\K$ \textit{under} $\omega$ if no adjacent vertex to $v$ has an adjacent vertex that uses the same encoding as $v$ and has the same color as $v$: 
    $ \forall u \in V(\K), \ \{v,u\} \in \K\implies \nexists x \in V(\K)\setminus\{v\} : \{u, x\} \in \K \land \omega(v) = \omega(x) \land \pi(x)=\pi(v)$. 
\end{definition}

\begin{definition}[Simplicial subcomplex distinguishability] \label{def:disting_simplex}
Let $\K$ be a simplicial complex, $\B \subseteq \K$ subcomplex, and $\omega : V(\K)\rightarrow E$. We say that $\B$ is distinguishable in $\K$ under $\omega$ if every vertex $v \in V(\B)$ is distinguishable in $\K$ under $\omega$.
\end{definition}

The primary challenge, and the key difference from~\cite{prequel}, is that algorithm $A$ can take an arbitrary number of rounds to obtain the same output complex as $r$ iterations of the unbounded full-information protocol. 
Figure~\ref{fig:iso-iter-diag} presents a commutative diagram that illustrates this concept. The bounded algorithm $A$ may require multiple rounds to achieve a protocol isomorphic to a single iteration of $\textit{FI}_P$.

\begin{figure}
\centering
\[\begin{tikzcd}
	{\Xi_{A(1)}(\mathcal{I})} & \dots & {\Xi_{A(k_1)}(\mathcal{I})} & \dots & {\Xi_{A(k_2)}(\mathcal{I})} & \dots & {\Xi_{A(k_r)}(\mathcal{I})} \\
	&&&&&&&& {\mathcal{O}} \\
	{\mathcal{I}} && {\Xi_{\textit{FI}_P} ( \mathcal{I})} && {\Xi_{\textit{FI}_P}^2 ( \mathcal{I})} & \dots & {\Xi_{\textit{FI}_P}^r(\mathcal{I})}
	\arrow[from=1-1, to=1-2]
	\arrow[from=1-2, to=1-3]
	\arrow[from=1-3, to=1-4]
	\arrow[from=1-4, to=1-5]
	\arrow[from=1-5, to=1-6]
	\arrow[dashed, tail reversed, from=1-5, to=3-5]
	\arrow[from=1-6, to=1-7]
	\arrow["\delta", maps to, from=1-7, to=2-9]
	\arrow[from=3-1, to=1-1]
	\arrow["{\Xi_{\textit{FI}_P}}"', from=3-1, to=3-3]
	\arrow[dashed, tail reversed, from=3-3, to=1-3]
	\arrow["{\Xi_{\textit{FI}_P}}"', from=3-3, to=3-5]
	\arrow["{\Xi_{\textit{FI}_P}}"', from=3-5, to=3-6]
	\arrow["{\Xi_{\textit{FI}_P}}"', from=3-6, to=3-7]
	\arrow[dashed, tail reversed, from=3-7, to=1-7]
	\arrow["\delta"', maps to, from=3-7, to=2-9]
\end{tikzcd}\]
\caption{Commutative diagram illustrating the equivalence between the unbounded full-information protocol map $\Xi_{\textit{FI}_P}$ and the bounded shared-memory iterated protocol map $\Xi_A$, using the write pattern $P$, over an input complex $\mathcal{I}$. $\Xi_A(k)$ denotes the associated protocol map at the k-th iteration of algorithm $A$. Dashed arrows indicate the existence of an isomorphism. The bounded algorithm $A$ is able to solve the same set of tasks as the unbounded full information protocol ($\delta \circ \Xi_A(\I) \cong \delta \circ \Xi_{\textit{FI}_P}(\I)$). }
\label{fig:iso-iter-diag}
\end{figure}
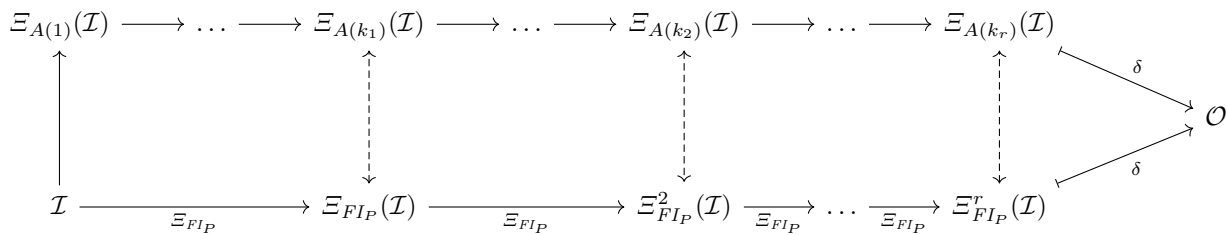

Theorem~\ref{theorem:encodings} gives a necessary condition for a bounded algorithm $A$ to be equivalent to its unbounded full information counterpart $\textit{FI}_P$. 
In a nutshell, each face of the input complex must be made distinguishable in $\I$ at some round of $A$.
To make a face $\sigma$ distinguishable in an input complex $\I$, the encoding function has to assign distinct values to each vertex in $\sigma$ so that, if the input configuration is $\sigma$, all processes can effectively map the observed values to the corresponding states in $\sigma$.
The encoding functions of $A$, $\omega_A$ must account for each global input configuration and convey the same information as in the full information protocol.
We consider each face $\sigma \in \text{Faces}(\I)$  as the subcomplex it induces in $\I$ when reasoning about distinguishability.


\begin{theorem} \label{theorem:encodings}
Let $\I$ be an input complex, $r>0$, $A \in \mathit{ITER}$ be an algorithm using the write-read pattern $P$, and $\{\omega_k\}_{k\leq r}$ the sequence of encoding functions of $A$.
Then:
\[
    \Xi_{A(r)}(\I) \cong \Xi_{\textit{FI}_P}(\I) \implies \forall \sigma \in \text{Faces}(\I)\ \exists k \leq r\ :\ \sigma\ \text{is distinguishable in}\ \I\ \text{under}\ \omega_{k}
\]
\end{theorem}
\begin{proof}
By contradiction, suppose that there exists a face $\sigma \in \I$ which is not distinguishable in $\I$ under any encoding function in $\{\omega_k\}_{k\leq r}$.
As $\Xi_{A(r)}(\I) \cong \Xi_{\textit{FI}_P}(\I)$, by Lemma~\ref{lemma:FIP_prop_cap} we have that for any subcomplex $\A\subseteq\I$, $\Xi_{A(k)}(\sigma \cap \A) = \Xi_{A(k)}(\sigma) \cap \Xi_{A(k)}(\A)$. 
We want to show that this requires that $\sigma$ is distinguishable in $\I$ under some $\omega_k \in \{\omega_k\}_{k\leq r}$.

Consider the case where $\I$ is a simplicial complex made of two simplices $\alpha$ and $\beta$ of dimension $n$ such that $\alpha \cap \beta = \sigma$ and $dim(\sigma)=n-1$. Note that both $\alpha$ and $\beta$ have a vertex which is not in $\sigma$. Let $a$ and $b$ be such vertices respectively and $p$ the respective process associated to both vertices: $p=\pi(a)=\pi(b)$.
Let $p$ write the same value for both states, that is: $\omega_k(a) = \omega_k(b)$. As $a$ and $b$ are adjacent to any vertex in $\sigma$, then $\alpha$ is not distinguishable in $\I$ under $\omega_k$. Analogously, $\beta$ is neither distinguishable in $\I$ under $\omega_k$.
By contradiction, suppose $\Xi_{A(k)}(\sigma) = \Xi_{A(k)}(\alpha) \cap \Xi_{A(k)}(\beta)$. Consider the execution where the global input configuration is $\alpha$ and a process $p'\neq p$ reads the value written by $p$ and $p$ reads nothing but its own input. Let $c_{p'}$ be the resulting final state of process $p'$. Because $\omega_k(a) = \omega_k(b)$, from the perspective of $p'$, the execution is \emph{indistinguishable} from the execution where process $p$ is in state $b$. 
Thus, the vertex $c_{p'}$ is both in $\Xi_{A(k)}(\alpha)$ and $\Xi_{A(k)}(\beta)$.
However, $a \notin \sigma$ and $b \notin \sigma$. Thus, $c_{p'} \notin \Xi_A(\sigma)$, a contradiction.

Therefore, there exists a $\omega_k \in \{\omega_k\}_{k\leq r}$ such that $\sigma$ is distinguishable in $\I$ under $\omega_k$, which contradicts the initial hypothesis. $\qed$

\end{proof}

Theorem~\ref{theorem:encodings} demonstrates that in an iterative protocol, each facet is made distinguishable at some round $k$ and remains so until the final round $R$. This is a defining property of the iterative protocol maps $\Xi_A$: they must subdivide the facets of the input complex at some step. Indeed, making a facet distinguishable is necessary for the protocol to make progress.

In~\cite{prequel}, a detailed analysis of input complex indistinguishability is provided. 
The authors introduce the concept of the \emph{indistinguishability graph} for an input complex and frame the challenge of making the complex distinguishable as a vertex coloring problem.
In this approach, a single encoding function is used to distinguish all faces of the input complex.
In our model, however, algorithm $A$ is equipped with a protocol map $\Xi_{A(k)}$ at each round, so rather than a single encoding function $\omega$ for the entire complex, we have a sequence of encoding functions $\{\omega_k\}_{k\leq R}$, each applied sequentially at each intermediate protocol complex.
Additionally, the protocol map preserves already distinguished faces.
Specifically, we seek a round $k$ such that $\Xi_{A(k)}(\I)$ distinguishes a face $\sigma$, and the output configurations persist through the subsequent protocol maps of $A$ until the last round.

To render a face distinguishable, the encoding function must assign unique values to some vertices in $\I$, and consequently, writing these values to shared memory requires a certain number of bits. For example, a single-bit encoding function can only toggle between writing 0 or 1, with memory initialized to 0 (or $\bot$), thereby allowing only only one (or at most two) distinct encodings. 

Formally, let $b$ denote the number of bits available in shared memory for each process entry per round in $A$. The number of possible encodings is then limited by this bit count: $|\Ima \omega_{A(k)}| \leq 2^b - 1$.

Thus, the task of designing a bounded iterative protocol $A$ that matches the unbounded full-information protocol’s behavior becomes a problem of identifying suitable encoding functions $\omega_{A(k)}$, constrained by $|\Ima \omega_{A(k)}| \leq 2^b - 1$, so that $A$ can distinguish all facets in $\I$ within as few rounds as possible. 
However, achieving this optimally is non-trivial: in fact, the problem is NP-hard~\cite{Karp-NP}.

First, we define the \textit{distinguishable simplicial subcomplex} of an encoding function $\omega$ as the subset of faces of an input complex $\I$ that are \textit{distinguishable} under $\omega$. 
Then, we can denote the distinguishable simplicial subcomplex of a sequence of encoding functions as follows:

\begin{definition} \label{def:disting_subcomplex}
Let $\omega_A$ be a sequence of encoding functions $\{\omega_{A(k)}\}_{0< k \leq R}$ and $\I$ an input complex. 
We define $\D(\omega_A,\I)$, the \emph{distinguishable subcomplex of $\I$ under $\omega_A$}, as follows:

$$
\D(\omega_A,\I) = \bigcup_{0 \leq k < R}{\{ \sigma \in \text{Faces}(\I) : \sigma\ \text{is distinguishable in}\ \I\ \text{under}\ \omega_{A(k)}\}}
$$
\end{definition}

Observe that $\D(\omega_A,\I)$ is a well-defined simplicial subcomplex, as if a face $\sigma \in \D(\omega_A,\I)$ is distinguishable under some $\omega_{A(k)}$, all its faces $\sigma' \subseteq \sigma $ are also distinguishable under the same encoding function. 
Using this definition we can restate Theorem~\ref{theorem:encodings} as follows: 

\begin{corollary} \label{corollary:distset}
Let $\I$ be an input complex, $A \in \mathit{ITER}$ using write-read pattern $P$, $\textit{FI}_P$ the unbounded full-information algorithm, $\omega_A$ the corresponding sequence of encoding functions of $A$ and $\D(\omega_A,\I)$ its distinguishable subcomplex. 
$$
\Xi_{A}(\I) \cong \Xi_{\textit{FI}_P}(\I) \Rightarrow \D(\omega_A,\I) = \I
$$
\end{corollary}

With these definitions and results, we can establish the computational hardness of finding a minimal encoding sequence, such that $\D(\omega_A,\I) = \I$ and algorithm $A$ takes as less rounds as possible for an arbitrary $\I$ to $\Xi_{A}(\I) \cong \Xi_{\textit{FI}_P}(\I)$.

\begin{lemma}\label{NP}
Let $b>0$, $\I$ be an input complex, $A \in \textit{ITER}$ taking $R$ rounds and $\{\omega_{A(k)}\}_{0< k \leq R}$ its corresponding sequence of encoding functions, where $|\Ima \omega_{A(k)}|\leq 2^b -1\ \forall 0 < k\leq R$ and $\D(\omega_A,\I) = \I$. 
Finding a sequence $\omega_A$ such that there is no shorter sequence satisfying the former conditions is NP-Hard.
\end{lemma}
\begin{proof}
We will prove this by showing that the problem can be used to solve an arbitrary instance of the Set Cover problem, therefore it has to be as hard as solving Set Cover~\cite{Karp-NP}.

Consider the family of encodings which their images is either $1$ or $\bot$. Let $A(\mathcal{I})$ be an algorithm that solves the search problem of the lemma, returning a valid cover with the respective encodings used.
Let $U$, and $S\subseteq
2^U$ be a set and a collection of subsets, respectively.
We build the complex $\mathcal{I}$ as follows: for each $u \in U$, define a facet $\Delta^{|U|}$ such that for each $s\in \overline{S}$ (complement of $S$), the facets with labels in $s$ share a common vertex. As we have $|U|$ facets, and each simplex could be attached to $|U|$ different simplices, $\mathcal{I}$ is well-defined.

Now $A(\mathcal{I})$ will return a minimal cover such that the facets of each subset will not share a vertex (as there couldn't be distinguishable by encodings writing only 1). Thus, the labels of the facets of each subset are in $S$. And each facet corresponds to an element $u \in U$. It follows that $A(\mathcal{I})$ solves the optimization version of Set Cover. $\qed$
\end{proof}

Given the memory constraint, obtaining the optimal cover is a hard task.
Consequently, our objective is to find the asymptotic bit complexity of the bounded algorithm $A$ with respect to its number of rounds, the desired $r$ iterations of the full-information protocol, and number of processes in the system $n$.

\subsection{Combinatorial Characterization of the Necessary Condition}

We now link the necessary condition of Lemma~\ref{lemma:lateResultLowerBound} to a combinatorial property of the input complex and subsequently present a result on its asymptotic growth with respect to the number of full-information rounds and available bits of shared-memory.

\begin{lemma} \label{lemma:min_rounds}
Let $\I$ be an input complex, and let $A \in \mathit{ITER}$ be an algorithm using the write-read pattern $P$ such that $\Xi_{A}(\I) \cong \Xi_{\textit{FI}_P}(\I)$.
Let $b$ denote the number of bits per entry in the iterated memory of $A$.
Then the number of iterations required by $A$ -- that is, the length of its sequence of encoding functions $\omega_A$ -- satisfies the following lower bound:

\[
|\omega_A| \geq 
\Bigg\lceil \frac{\max_{v \in V(\I)}
{\mathrm{deg}(\I,v)}}
{n (2^b - 1)}
 \Bigg\rceil
\]

\end{lemma}
\begin{proof} 
Let $v \in V(\I)$ be a vertex of maximum degree $D$ in the input complex $\I$, corresponding to some process $p$.
For process $p$ in state $v$ to distinguish each of the $D$ states adjacent to $v$ in $\I$, each such state must be assigned a different encoding value by the encoding functions in $\omega_{A}$.
Since each encoding function $\omega_{A(k)}$ is constrained by the $b$ bits available in shared memory, it can assign at most $2^b - 1$ distinct values to the vertices of a given process, thereby distinguishing at most $2^b - 1$ states.

To achieve this, each $\omega_{A(k)}$ can assign a unique value to at most $2^b-1$ of the vertices adjacent to $v$, and assign $\bot$ to all others. 
Therefore, at least $\lceil \frac{D}{2^b-1} \rceil$ different encoding functions are required to make the distinguishable subcomplex $\D(\omega_A,\I)=\I$, a necessary condition for $\Xi_{A(r)}(\I) \cong \Xi(\I)$ by Corollary~\ref{corollary:distset}.
Finally, the number of vertices of color $p$ adjacent to $v$ have to be as big as counting all the vertices of all colors divided by $n$, which is the case when $\mathrm{Lk}(\I,v)$ has the same number of vertices of every color. $\qed$
\end{proof}

Lemma~\ref{lemma:min_rounds} gives a topological condition on the number of rounds required by $A$ to yield a full-information protocol. 
It is of interest understanding how the value grows asymptotically. To achieve this, we will take advantage of the results and tools provided in~\cite{prequel}. 
We recommend the reader to check~\cite{prequel} if it is interested in the tools used, such as the $\mathrm{f}$-vector, a common tool in polyhedral combinatorics~\cite{fvector}, to perform the asymptotic analysis. 

First, we show that the degree of a vertex in the input complex, after $r$ iterations of the full-information protocol map, will be bigger than applying the standard chromatic subdivision $r$ times.  

\begin{lemma} \label{lemma:lateResultLowerBound}
Let $\I$ be an input complex, $\textit{FI}_P$ the unbounded full information protocol, $v \in V(\mathcal{A})$, $r > 0$, and $\mathrm{Ch}$ the standard chromatic subdivision operator. Then the following inequality holds:

\[
{\mathrm{deg}(\Xi_{\textit{FI}_P}^r(\I),v)} \geq
{\mathrm{deg}(\mathrm{Ch}^r(\I),v)}
\]
\end{lemma}
\begin{proof}
We will prove the statement by induction over $r$. 
Let $p$ be a process such that $\pi(v)=p$.
First, note that $\Xi_{\textit{FI}_P}$ is a full information protocol, thus for each face $\sigma$ such that $dim(\sigma)=n$ and $v \in \sigma$, the protocol map yields $(n-1)$ central vertices in $\sigma$. 
Moreover, there is an edge ($1$-face) connecting vertex $v$ and each of the new vertices in $\mathrm{Ch}(\I)$ of the faces including $v$, as it represents the configuration where another process $p' \neq p$ reads the input value of $p$ and $p$ only reads its own input. 
Such configuration does exist in any algorithm $\I \in \mathit{ITER}$ when process $p$ executes all its steps before all the other processes. 
Moreover, by Lemma~\ref{lemma:FIP_prop_mesh}, any edge adjacent to $v$ in $\I$ does not appear in $\Xi_{\textit{FI}_P}(\I)$ nor in $\mathrm{Ch}(\I)$.
Observe that enumerating these vertices gives exactly ${\mathrm{deg}(\mathrm{Ch} (\I),v)}$. 
Thus, the degree cannot be smaller than the one of the chromatic subdivision. 
For the inductive step we have that ${\mathrm{deg}(\Xi_{\textit{FI}_P}^r(\I),v)} \geq
{\mathrm{deg}(\mathrm{Ch}^r(\I),v)}$. 
It suffices to show that the new edges generated in the $(r+1)$ subdivision correspond to valid executions in $\textit{FI}_P$, since, by Lemma~\ref{lemma:FIP_prop_mesh}, neither $\Xi_{\textit{FI}_P}^{r+1}(\I)$ nor $\mathrm{Ch}^{r+1}(\I)$ contains edges which are also in $\Xi_{\textit{FI}_P}^{r}(\I)$ or $\mathrm{Ch}^{r}(\I)$.
As in the base step, the new edges correspond to the central vertices yielded at the subdivision of each face $\sigma : v \in \sigma$.
The edges between these vertices and $v$ represent valid configurations of the full information protocol. 
Therefore, ${\mathrm{deg}(\Xi_{\textit{FI}_P}^{r+1}(\I),v)}$ has to be bigger or equal than the degree at $(r+1)$ chromatic subdivisions. $\qed$
\end{proof}

By applying the asymptotic bound from Lemma~\ref{lemma:asymp_boundf} to the construction in Lemma~\ref{lemma:lateResultLowerBound}, we derive a lower bound on the round complexity as a function of the number of available bits $b$.

\begin{theorem}
\label{theorem:final_low}
Let $\I$ be an $n$-dimensional input complex, with $n>2$ and let $r>0$. 
Let $\textit{FI}_P$ denote the unbounded full information protocol that uses a write-read pattern $P$.
Let $A \in \mathit{ITER}$ be an algorithm that uses a write-read pattern $P$, but uses at most $b$ bits per process per round.
Any such algorithm $A$ satisfying $\Xi_{A}(\I) \cong \Xi_{\textit{FI}_P}^r(\I)$ must perform $\Omega\Bigl({(n!)^{r-1} \cdot 2^{n-b}} \Bigr)$ rounds.
\end{theorem}

\begin{corollary}
Under the same context of Theorem~\ref{theorem:final_low}, if we want $A$ to perform a single round per iteration of the full-information protocol, that is $|\omega_A|=1$, then $\Omega( r n \log n)$ is a lower bound on the bit complexity of $A$. Conversely, if we want $\Xi_{A}(\I) \cong \Xi_{\textit{FI}_P}^r(\I)$ using $2$ bits, we have that $\Omega((n!)^{r-1} 2^n)$ is a lower bound on the round complexity of $A$. 
\end{corollary}

\section{Greedy Star: An Asymptotically Optimal Algorithm} 
\label{sec:greedy_star}

In this section, we present Algorithm~\ref{alg:greedyCover}, an algorithm that leverages the inequality of Lemma~\ref{lemma:min_rounds} to construct a sequence of encoding functions $\omega_S$.
We devise an algorithm to construct a sequence of encoding functions that yield \emph{distinguishable subcomplexes} forming a cover of the input complex.
We call this algorithm Greedy Star, as it proceeds by iteratively selecting vertices in $\I$ and ensuring that each selected vertex is distinguishable by assigning it a unique encoding within its star.
To maintain distinguishability, the algorithm enforces that the star of each selected vertex has an empty intersection with the stars of previously selected vertices within an encoding function.

\begin{algorithm}
\SetAlgoLined
\KwIn{Chromatic input complex $\mathcal{I}$.}
\KwOut{Sequence of encoding functions $\{\omega_r\}_{r\geq 0}$ such that $\D(\{\omega_r\}_{r\geq 0},\I)=\I$.}

$\mathcal{A} \gets \emptyset$; $V \gets V(\mathcal{I})$; $W \gets \emptyset$; $r \gets 0$ \\
\While{$\mathcal{A}\neq\mathcal{I}$}{
    \textbf{let} $\omega_r : V(\I) \rightarrow \mathbb{N}\cup\{\bot\},\ w_r(x) = \bot$\\
    $\mathcal{U} \gets \emptyset$ \\
    $V_u \gets \emptyset$ \\
    \ForEach{$v \in V$}{
        $\mathcal{S} \gets \mathrm{St}(\mathcal{I},v)$ \\
        \If{$\forall w \in V(\mathcal{S}), \mathrm{St}(\mathcal{I},w) \cap \mathcal{U} = \emptyset$}{
            $\mathcal{U}\gets\mathcal{U} \cup \mathcal{S}$ \\
            $V_u \gets V_u \cup V(\mathcal{S})$ \\ 
            \tcp{redefine $\omega_r$ for vertices in the star.}
            \textbf{let} $\forall x \in V(\mathcal{S}), k \in \mathbb{N},\ \omega_r(x)=k : \nexists y \in V(\mathcal{S}),\ x\neq y \land \omega_r(y)=k \land \pi(y)=\pi(x)$  \\
        }
    }
    $\mathcal{A}\gets\mathcal{A}\cup\ \mathcal{U}$\\
    $V \gets V\setminus V_u$ \\
    $W \gets W \cup \{\omega_r\}$ \\
    $r \gets r+1$ \\
}
\Return $W$
\caption{Greedy Star encoding algorithm: $GS(\mathcal{I})$}
\label{alg:greedyCover}
\end{algorithm}

\begin{lemma}
The Greedy Star Cover algorithm outputs a sequence of encoding functions such that $D(\omega_S(\I),\I) = \I$.
\end{lemma}
\begin{proof}
First, note that the while condition is not satisfied until $\mathcal{A}=\I$, so we have to show that the condition is eventually reached.
Observe that in each for loop we start with  $\mathcal{U}=\emptyset$, so on each iteration a new star subcomplex will be added to $\mathcal{A}$.
Note that $V$ is initially the set of vertices of $\mathcal{I}$, and vertices are only removed when they are used to generate a new star. 
As $\bigcup_{v \in V(\mathcal{I})}{St(v)}=\mathcal{I}$, after a finite number of iterations we have that $\A = \I$.
Each star $\mathrm{St}(v)$ added to the cover corresponds to the image of the input complex under one encoding function $\omega_S(k)$, and is made distinguishable from the others by construction (since the algorithm ensures that these stars are disjoint in the encoding space by assigning unique encodings).
Hence, for every face $\sigma \in \I$, there exists some $k$ such that $\sigma$ lies in a star $\mathrm{St}(v)$ added at iteration $k$ and is distinguishable in $\I$ under $\omega_S(k)$. Therefore, \(\sigma \in D(\omega_S(k), \I)\), and so \(\sigma \in D(\omega_S, \I)\). $\qed$
\end{proof}

Observe that each encoding function returned by Greedy Star requires a number of distinct encodings proportional to the maximum number of facets in the star of some vertex.
However, the number of available encodings is bounded: $|\operatorname{Im} \omega_{A(k)}| \leq 2^b - 1$.
Therefore, a bounded algorithm $A \in \textit{ITER}$ must simulate each encoding function $\omega_S(k)$ over multiple iterations in order to assign unique encodings to all the facets in a star.
In particular, each $\omega_S(k)$ must be further decomposed into multiple encoding functions, ensuring that all facets in the corresponding star remain distinguishable across these iterations.
Furthermore, any process $p$ at a vertex $v$ such that $\omega_S(k)(v) = \bot$ ignores the contents of shared memory during such iteration and preserves its local state.

\begin{lemma}\label{lemma:max_rounds_star}
Let $\mathcal{I}$ be an input complex, $b>0$, and let $\omega_{S}(\I)$ be the sequence of encoding functions constructed by the Greedy Star algorithm such that for each $\omega_k \in \omega_{S}(\I)$, $|\operatorname{Im} \omega_k| \leq 2^b -1$.
\[
|\omega_S| \leq 
 4\Bigg\lceil \frac{\max_{v \in V(\I)}
{\mathrm{deg}(\I,v)}}
{n (2^b - 1)}
 \Bigg\rceil
\]
\end{lemma}
\begin{proof}
At each iteration of the outer while loop, the algorithm constructs an intermediate set $\mathcal{U}$ of disjoint vertex stars, each centered at a vertex $v$, and assigns a unique encoding to each facet in these stars.
Since each encoding function can assign at most $2^b - 1$ distinct values (excluding $\bot$), and each vertex star can include up to $\deg(\mathcal{I}, v)$ facets, it follows that in order to assign unique values to all such facets, we need at most
$
M := \Bigg\lceil \frac{\max_{v \in V(\mathcal{I})} \deg(\mathcal{I}, v)}{n (2^b - 1)} \Bigg\rceil
$
encoding functions per iteration of the while loop.

Let $\mathcal{A}$ be the set of already encoded (i.e., covered) facets, and $\overline{\mathcal{A}} = \mathcal{I} \setminus \mathcal{A}$ be the remaining uncovered subcomplex.
We now show that after at most $4$ iterations of the while loop, the entire complex $\mathcal{I}$ is covered. The key insight is that $\I$ shrinks 
across iterations.

In the first iteration, all selected stars are disjoint by construction, and their union covers a subset of the complex.
Any vertex in $\overline{\mathcal{A}}$ that is adjacent (via a facet) to some vertex in $\mathcal{A}$ is at most distance $3$ edges away from a fully encoded facet. 
If not, the intersection of the star of one of the two vertices constituting the path will be empty, so the star of such vertex should be included be in $\mathcal{U}$.
In the second iteration, the new stars are chosen around such frontier vertices. Because these are now adjacent to already covered regions, their stars include facets that were partially covered and are now fully encoded.
By repeating this reasoning, in the third iteration, the uncovered subcomplex $\overline{\mathcal{A}}$ is reduced to isolated facets or vertices surrounded by covered stars.
In the fourth iteration, any remaining uncovered facet must intersect a fully covered neighborhood, implying its star is now fully included in the encoding.
Thus, after at most $4$ iterations of the while loop, all facets in $\mathcal{I}$ have been assigned encodings. $\qed$
\end{proof}

As a result, we can construct algorithms $\A \in \mathit{ITER}$ using Greedy Star to get their encoding function $\omega_A$. 
As the cover is constructed deterministically, it can be computed locally by each process. 
Thus, each process shares the same sequence $\omega_A$ of encoding functions.
Then, the next state function $\sigma_{A(k)}$ will map the read values to the respective processes' states by inverting the encodings of $\omega_{A(k)}$. 
Note that by definition of distinguishability, such an inverse exists, as all $\omega_{A(k)}$ are injective. 
In each round, processes accumulate information about the states of other processes by writing to shared memory, provided that the cover subcomplex for that round includes the vertex representing their current state.

In the final round $R$, each process gains the same knowledge as it would in a single iteration of the full-information protocol. 
Consequently, the terminal states in the protocol complexes  $\Xi_A$  and  $\Xi_{\textit{FI}_P}$ are effectively identical. 
This protocol complex can then be used as input for a subsequent iteration of the full information protocol.

Furthermore, the generic decision function  $\delta$ used in the full-information protocol can seamlessly be applied to $A$ as well, since the final states of each process in $A$  are isomorphic to those in the full-information protocol.

It is important to remark that the Greedy Star algorithm is agnostic to the underlying read-write communication pattern used by an algorithm $A \in \textit{ITER}$.
In particular, using stronger write-read patterns such as atomic or immediate snapshots with Greedy Star, instead of the iterated collect model, offers no advantage: the resulting protocol map collapses to that of iterated collect.
%
This is due to the round partitioning imposed by Greedy Star, which isolates communication to disjoint regions of the input complex at each round.
Nonetheless, we show that an algorithm using the iterated collect pattern, along with the encoding functions produced by Greedy Star, faithfully simulates the unbounded full-information iterated collect protocol.
Thus, Greedy Star yields an asymptotically optimal protocol.

\begin{theorem} \label{theorem:gs_gives_collect}
Let $\mathcal{I}$ be an input complex, and let $GS \in \textit{ITER}$ be an algorithm that uses the same read-write pattern as Iterated Collect, with a sequence of encoding functions $\omega_{GS} = GS(\mathcal{I})$ produced by the Greedy Star algorithm. 
Then:
\[
\Xi_{GS}(\mathcal{I}) \cong \Xi_{\textit{FI}_{IC}}(\mathcal{I}),
\]
\end{theorem}
\begin{proof}
    First, we will show that $ \Xi_{GS}(\I) \subseteq \Xi_{\textit{FI}_{IC}}(\I)$. We want to show that any global configuration of ${GS}(\I) $ is also a valid global configuration in $\textit{FI}_{IC}$.
    By the pseudocode of Algorithm~\ref{alg:IterMemory_algo}, there is an execution of GS where every process writes and reads only its own input.
    Such execution is present at $\Xi_{\textit{FI}_{IC}}(\I)$.
    Now consider executions where processes may read, at some round in $GS$, the inputs of additional processes.
    Since processes retain all previously read inputs, it suffices to show that any execution of iterated collect can be extended by allowing a process to read additional inputs.
    Consider an execution $\alpha$ of iterated collect and a process $p$ which has not read the input of some other process $q$ in $\alpha$.
    We can construct an execution $\alpha'$ identical to $\alpha$ except that $p$ additionally reads the input of $q$.
    It suffices to rearrange the read operation of $p$ on the register of process $q$ in $\alpha$ such that it occurs afterwards the write operation by $q$.
    Note that the views of all other processes remain unchanged.
    Therefore, any global configuration reachable in $GS$ is also reachable in $\textit{FI}_{IC}$, and we have $\Xi_{GS}(\I) \subseteq \Xi_{\textit{FI}_{IC}}(\I)$. 

    It remains to show that $\Xi_{\textit{FI}_{IC}} (\I) \subseteq \Xi_{GS}(\I)$.
    Consider a final configuration $\sigma \in \Xi_{\textit{FI}_{IC}} (\I)$. 
    Because $GS$ covers the entire input complex, there exists a round $k$ in $GS$ where the input configuration is distinguishable under $\omega_{GS}(k)$.
    Hence, there exists an execution $\beta$ in $GS$ such that the view of all process at round $k$ is exactly the same as in $\sigma$.
    We now construct an extension $\beta'$ of $
    \beta$ in which all processes terminate while preserving the same final configuration $\sigma$.
    Because the model allows arbitrary asynchrony, we can extend the execution by scheduling all remaining steps of processes such that each read after round $k$ returns a value already observed by the reading process. Consequently, no process obtains new inputs beyond what it had at round $k$.
    Thus, no process gains additional information and their local state remains unchanged
    Thus, the final configuration of $\beta'$ is equal to $\sigma$.
    Therefore, $\Xi_{\textit{FI}_{IC}}(\I) \subseteq \Xi_{GS}(\I)$ and $\Xi_{\textit{FI}_{IC}}(\I) \cong \Xi_{GS}(\I)$. $\qed$
\end{proof}

From a bounded implementation of the iterated collect model, we can derive bounded algorithms with stronger communication patterns, such as atomic and immediate snapshots.
In particular, Algorithm~$5$ from~\cite{delporte2024computational} simulates a single round of immediate snapshot using $n$ rounds of iterated collect.
This simulation can be applied in a black-box manner to our bounded iterated collect implementation, yielding bounded full-information algorithms for snapshot-based models.
As a result, if the bounded iterated collect algorithm terminates in $R$ rounds, the corresponding derived algorithm terminates in $O(n \cdot R)$ rounds. Therefore, the asymptotic round complexity is preserved up to a linear factor in the number of processes.

\section{Final Results and Conclusions}
\label{sec:conclusions}

Combining Theorem~\ref{theorem:final_low} and Theorem~\ref{theorem:gs_gives_collect}, we derive asymptotic bounds that characterize the fundamental trade-off between the bit complexity ($b$ bits available per process per round) and the round complexity (i.e., the number of rounds) of bounded full-information algorithms.

\begin{theorem}[Tight bound on the bounded full-information iterated collect protocol] \label{theorem:finalbound_IC}
Let $\I$ be an $(n-1)$-dimensional input complex, with $n>2$ and let $r>0$. 
Let $\mathit{FI}_{IC}$ denote the \emph{unbounded full-information protocol} based on \emph{iterated collect}.
Let $A \in \mathit{ITER}$ be an algorithm that uses the same read-write pattern as iterated collect, but uses at most $b$ bits per process per iteration.
Let $|\omega_A|$ denote the number of encoding functions used by A, i.e., the number of rounds performed.
Then there exists such an algorithm $A$ satisfying $\Xi_{A}(\I) \cong \Xi_{\textit{FI}_{IC}}^r(\I)$, and the number of rounds satisfies: 
$$|\omega_A| \in \Theta\Bigl( (n!)^{r-1} \cdot 2^{n-b}\Bigr)$$
Moreover, no such algorithm exists that uses asymptotically fewer rounds. 
\end{theorem}

\begin{theorem}[Bounds on the bounded full-information iterated snapshot protocols]  \label{theorem:finalbound_Snap}

Let $\I$ be an $(n-1)$-dimensional input complex, with $n>2$ and let $r>0$. 
Let $\mathit{FI}_{S}$ denote the \emph{unbounded full-information protocol} that uses atomic or immediate snapshots as its read-write pattern $S$ ($S\in\{\textit{IAS},\textit{IIS}\}$). 
Let $A \in \mathit{ITER}$ be an algorithm that uses the same read-write pattern as $\mathit{FI}_{S}$, but with at most $b$ bits per process per iteration.
Let $|\omega_A|$ denote the number of encoding functions used by A, i.e., the number of rounds performed.
Then any such algorithm $A$ satisfying $\Xi_{A}(\I) \cong \Xi_{\textit{FI}_{S}}^r(\I)$ must perform at least
$
|\omega_A| \in \Omega\bigl((n!)^{r-1} \cdot 2^{n - b}\bigr)
$
rounds, and there exists such an algorithm that performs at most
$
|\omega_A| \in O\bigl((n!)^{r-1} \cdot 2^{n - b} \cdot n\bigr)
$ rounds.
\end{theorem}

The Greedy Star algorithm introduced in this work provides what appears to be the first efficient construction for simulating unbounded full-information protocols using a bounded number of bits per round.
Notably, our results apply to several iterated memory models studied in the literature, including iterated collect, atomic snapshot, and immediate snapshot.
This contribution highlights the power of combinatorial topology not only as an analytical tool for distributed systems, but also as a constructive tool for designing algorithms that exploit the topological structure of computation.

Several research directions remain open. One promising avenue is to adapt the techniques developed here to problems beyond full-information, such as approximate agreement~\cite{distCompTopo}.
Additionally, an important open question is whether the upper bound for bounded implementations of snapshot-based models can be improved. 
Specifically, given $b$-bit snapshot entries per process per round, it remains to determine whether an unbounded full-information snapshot protocol can be implemented using fewer than $O\bigl((n!)^{r-1} \cdot 2^{n - b} \cdot n\bigr)$ rounds.
Finally, one can also consider this question in the context of bounded-memory \emph{adversarial} models~\cite{HR10,KRH18}. 

\subsubsection{Acknowledgments.}

This work was supported by CHIST-ERA
(grant ANR-23-CHR4-0009).

\bibliographystyle{splncs04}
\bibliography{biblio,references}{}

\appendix

\end{document}